\documentclass[a4paper,11pt]{article}
\usepackage{amsmath,amssymb,amsthm,fullpage}

\newtheorem{thm}{Theorem}
\newtheorem{lem}[thm]{Lemma}
\newtheorem{cor}[thm]{Corollary}

\theoremstyle{definition}
\newtheorem{defn}[thm]{Definition}
\newtheorem{exa}[thm]{Example}

\theoremstyle{remark}
\newtheorem{rem}[thm]{Remark}

\newcommand{\mcC}{\mathcal C}
\newcommand{\mcD}{\mathcal D}
\newcommand{\mcI}{\mathcal I}
\newcommand{\mcJ}{\mathcal J}
\newcommand{\mcK}{\mathcal K}
\newcommand{\mcN}{\mathcal N}

\newcommand{\N}{\mathbb N}
\newcommand{\R}{\mathbb R}
\newcommand{\GF}{\operatorname{GF}}
\newcommand{\supp}{\operatorname{supp}}
\newcommand{\rank}{\operatorname{rank}}

\title{{\bf New Results on the Pseudoredundancy}\thanks{%
   Part of the paper has been presented at the ISIT, 2014, Hawaii, U.S..
    This work was supported by The National Science Foundation of China
    (No.~11171366).}}

\author{%
  {Zihui Liu}\thanks{%
  Corresponding author: lzhui@bit.edu.cn}\\
  {\normalsize Department of Mathematics,}\\
  {\normalsize Beijing Institute of Technology,
    Beijing 100081, China}\medskip\\
  {Jens Zumbr\"agel}\\
  {\normalsize Institute of Algebra,}\\
  {\normalsize Dresden University of Technology,
    01062 Dresden, Germany}\medskip\\
  {Marcus Greferath}\\
  {\normalsize Claude Shannon Institute,}\\
  {\normalsize University College Dublin, Belfield,
    Dublin 4, Ireland}\medskip\\
  {Xin-Wen  Wu}\\
  {\normalsize School of Information and Communication Technology,}\\
  {\normalsize Griffith University, Gold Coast, QLD 4222, Australia}}

\date{}

\begin{document}

\maketitle


\begin{abstract}\noindent
  The concepts of pseudocodeword and pseudoweight play a fundamental
  role in the finite-length analysis of LDPC codes.  The
  pseudoredundancy of a binary linear code is defined as the minimum
  number of rows in a parity-check matrix such that the corresponding
  minimum pseudoweight equals its minimum Hamming distance.  By using
  the value assignment of Chen and Kl\o ve we present new results on
  the pseudocodeword redundancy of binary linear codes.  In
  particular, we give several upper bounds on the pseudoredundancies of
  certain codes with repeated and added coordinates and of certain
  shortened subcodes.  We also investigate several kinds of
  $k$-dimensional binary codes and compute their exact pseudocodeword
  redundancy.  \medskip

  \medskip

\end{abstract}


\section{Introduction}

The concept of a pseudocodeword plays a key role in the finite-length
analysis of binary low-density parity-check (LDPC) codes under linear
programming (LP) decoding (or, to some extent, under message-passing
iterative decoding), see~\cite{fwk,vk}.  The effect of pseudocodewords
on the decoding behavior is measured by their pseudoweight~\cite{fkkr,vk},
which depends on the channel at hand. Accordingly, the pseudocodeword
redundancy (or pseudoredundancy) of a binary linear code is of
interest, which is defined as the minimum number of rows in a
parity-check matrix such that the corresponding minimum pseudoweight
is as large as its minimum Hamming distance.  The pseudoredundancy for
various channels has been studied, e.g., in~\cite{sv}, \cite{ks},
and~\cite{zsf}.

If a code has infinite pseudocodeword redundancy, then LP decoding for
this code can never achieve the maximum-likelihood decoding
performance; on the other hand, if a code's pseudocodeword redundancy
is finite, its value provides an indication of the required LP
decoding complexity.  Note that this is a fundamental complexity
associated with the code, and not tied to a particular parity-check
matrix.

It is undoubtedly meaningful to determine either the pseudocodeword
redundancy or to give bounds on the pseudocodeword redundancy of a
binary linear code.  However, it was shown in
\cite[Th.~3.2,~Th.~3.5]{zsf} that most codes have infinite AWGNC and
BSC pseudoredundancy.  In contrast to this result, we will determine
the pseudoredundancies of some kinds of $k$-dimensional codes and give
bounds for certain constructed codes.  Our main tool to study the
pseudoredundancy is the value assignment introduced by Chen and Kl\o
ve~\cite{ck}.

The rest of the paper is organized as follows.  In Section~\ref{sec:prelim} we
define pseudoweights for various channels and the notion of
pseudoredundancy; we also present the concept of value assignment.
Section~\ref{sec:repeat} contains a discussion of codes based on
repeating or adding coordinates and of shortened subcodes.  In
Section~\ref{sec:k-dim} we determine the pseudoredundancies of certain
$k$-dimensional codes based on value assignment, generalising previous
results significantly.  Finally, we conclude in
Section~\ref{sec:concl}.


\section{Preliminaries}\label{sec:prelim}

For a binary linear code $\mcC$ of length $n$, when analyzing LP
decoding for a binary-input output-symmetric channel, one may assume
that the zero codeword~$\mathbf{0}$ has been sent; then, the
probability of correct LP decoding depends on the conic hull of the
fundamental polytope, called the {\em fundamental cone}~\cite{fwk,vk},
which depends on the given parity-check matrix of~$\mcC$.

Let $H$ be an $m\times n$ parity-check matrix for $\mcC$, where the
$m$ rows may be linearly dependent.  Let $\mcI = \{1,\dots,n\}$ and
$\mcJ = \{1,\dots,m\}$ be the set of column and row indices,
respectively, and for each $j\in\mcJ$ let $\mcI_j = \{i\in\mcI \mid
H_{j,i}\ne 0\}$.  Then, the fundamental cone $\mcK(H)$ with respect to
the parity-check matrix $H$ of $\mcC$ is given as the set of vectors
$x\in\R^n$ that satisfy
\begin{equation}\label{eq:cone}
  \begin{split}
  \forall j\in\mcJ\ \forall \ell\in\mcI_j &:\
  x_{\ell}\le \textstyle\sum_{i\in\mcI_j\backslash\{\ell\}}x_i \,, \\
  \forall i\in\mcI &:\ x_i\ge 0 \,.
  \end{split}
\end{equation}
The vectors $x\in\mcK(H)$ are called {\em pseudocodewords} of~$\mcC$
with respect to the parity-check matrix~$H$.

The influence of a nonzero pseudocodeword on the decoding performance
is measured by its {\em pseudoweight}, which depends on the underlying
channel.  The BEC (binary erasure channel), AWGNC, BSC pseudoweights
and max-fractional weight of a nonzero pseudocodeword $x\in\mcK(H)$
are defined as follows~\cite{fkkr,vk}:
\begin{gather*}
  w_{\text{BEC}}(x) = |\text{supp}(x)| \,,\\
  w_{\text{AWGNC}}(x) = \frac{(\textstyle\sum_{i\in\mcI}x_i)^2}
  {\textstyle\sum_{i\in\mcI}{x_i}^2} \,,
\end{gather*}
letting $x'$ be a vector in $\R^n$ with the same components as $x$ but
in nonincreasing order,  for $i-1<\xi\leq i$, where $1\leq i\leq n$,
letting $\phi(\xi)=x'_i$ and defining
$\Phi(\xi)=\int_{0}^{\xi}\phi(\xi')d\xi'$,
\begin{gather*}
  w_{\text{BSC}}(x)=2\Phi^{-1}(\Phi(n)/2) \,, \\
  w_{\text{maxfrac}}(x) = \frac{\textstyle\sum_{i\in\mcI}x_i}
  {\textstyle\max_{i\in\mcI}x_i} \,.
\end{gather*}

For binary vectors $x\in\{0,1\}^n\setminus\{\mathbf{0}\}$ one has
\[ w_{\text{BEC}}(x) \,=\, w_{\text{AWGNC}}(x) \,=\,
w_{\text{BSC}}(x) \,=\, w_{\text{maxfrac}}(x) \,=\, w_H(x) \,, \] where
$w_H(x)$ denotes the Hamming weight of $x$.

Define the {\em minimum pseudoweight} of a code $\mcC$ with respect
to a parity-check matrix $H$ as
\[ w_{\min}(H) \,= \! \min_{x\in\mcK(H)\setminus\{\mathbf{0}\}}\! w(x) \,, \]
where $w(x)$ may represent any one of the four pseudoweights (it is
a fact that $w_{\min}(H)$ is indeed attained on
$\mcK(H)\setminus\{\mathbf{0}\}$~\cite{vk}).  The minimum
pseudoweight $w_{\min}(H)$ can be seen as a first-order measure of
decoding error-correcting performance of a code $\mcC$ given by the
parity-check matrix $H$ under LP decoding.  We note that all four
minimum pseudoweights are upper bounded by $d(\mcC)$, the minimum
distance of $\mcC$.

\begin{defn}\label{def:pr}
  The {\em pseudocodeword redundancy}, or briefly the {\em
    pseudoredundancy}, $\rho(\mcC)$, of a binary linear $[n,k,d]$ code
  $\mcC$ is defined as
  \[ \rho(\mcC) \,=\, \inf \{ \#\text{rows}(H) \mid
  H \text{ is a parity-check matrix of }\mcC,w_{\min}(H) = d\} \,, \]
  where $\inf \varnothing$ is defined as $\infty$; here $w_{\min}(H)$ is
  for one of the four pseudoweights, and we use accordingly the term
  BEC, AWGN, BSC, or max-fractional pseudoredundancy.
\end{defn}

It is obvious that $\rho(\mcC)\ge n-k$ for any $[n,k,d]$ code $\mcC$.
Furthermore, for any binary linear code $\mcC$ it
holds~\cite[Th.~2.5]{zsf} that
\begin{equation}\label{eq:inequality}
 \begin{array}{c}
  \rho_{\text{maxfrac}} (\mcC) \;\ge\; \rho_{\text{AWGNC}} (\mcC) \;\ge\;
  \rho_{\text{BEC}} (\mcC) \; , \\
  \rho_{\text{maxfrac}} (\mcC) \;\ge\; \rho_{\text{BSC}} (\mcC) \;\ge\;
  \rho_{\text{BEC}} (\mcC) \; .
 \end{array}
\end{equation}

The value assignment, which was first introduced in~\cite{ck},
is our main tool for investigating the pseudoredundancy.  It is given
as follows.

\begin{defn}
  A {\em value assignment} is a map
  \[ m(\cdot): PG(k\!-\!1,q) \,\to\, \N = \{0,1,2,\dots,\}  \]
  from the $(k\!-\!1)$-dimensional projective space $PG(k\!-\!1,q)$
  over the finite field $\GF(q)$ to $\N$, the set of nonnegative integers.  For
  a point $p\in PG(k\!-\!1,q)$, we call $m(p)$ the {\em value} of $p$.
\end{defn}

For instance, let $G$ be a
$k\times n$ matrix over $GF(q)$ (which may be a generator matrix of a linear $[n,~k]$ code);
the columns of $G$ are viewed as points of $PG(k\!-\!1,q)$. For any $p\in PG(k\!-\!1,q)$, we
define $m(p)$ to be the number of occurrences of $p$ as columns of $G$. Please
note that $G$ may have repeating columns, for which the corresponding point $p$
of $PG(k\!-\!1,q)$ has a value greater than 1. If $p$ does
not appear in $G$, then $m(p)=0$.
Viewing the columns of a generator matrix of a linear $[n, k]$ code
as a multiset of points of the projective space $PG(k\!-\!1,q)$, this multiset defines a
value assignment. 
Conversely, given a value assignment $m$ (or equivalently, a
sequence of nonnegative integers, $z_1, z_2, ..., z_N$, where
$z_i=m(p_i)$ for $p_i\in PG(k\!-\!1,q)$ and $N=(q^k-1)/(q-1)$ the
number of points of $PG(k\!-\!1,q)$), a $k\times n$ matrix $G$ is
uniquely determined, where $n$ is the number of  points $p$ (each
$p$ is counted $m(p)$ times),
of the projective space $PG(k\!-\!1,q)$.  
The columns of $G$
is a multiset of points of $PG(k\!-\!1,q)$, that is, the
columns of $G$ consist of the points $p$ with
positive values and each of them repeats $m(p)$ times. Therefore, the value assignment defines a generator matrix $G$ and thus an
$[n, k]$ code (up to code equivalence).
%
Let $\mcC$ be the $[n,k]$ code determined by a value
assignment $m(\cdot)$, from the above discussion the following important property holds, $\sum_{p\in PG(k\!-\!1,q)}m(p) = n$.

\vskip 0.2cm

Since equivalent codes lead to equivalent dual codes, equivalent codes
with parity-check matrices have the same minimum pseudoweight.  Thus,
when studying the pseudocodeword redundancy, it suffices to use the
value assignment to construct different equivalent codes.


\section{Codes Based on Repeating and Adding Coordinates
  and Shortened Subcodes}\label{sec:repeat}

In this section, we will give several bounds on the pseudoredundancies
of codes obtained by increasing the number of coordinates and of
certain shortened subcodes by using the value assignment.

We remark that repeating coordinates is a useful method to construct
a code.  For example, any binary linear constant-weight code (that
is, all the nonzero codewords have the same weight) is obtained from
a simplex code by repeating each coordinate equally times~\cite{lc},
or equivalently,  any binary
  constant-weight code consists of copies of the simplex code.
  Furthermore, a recent paper~\cite{lw} shows that
a large class of codes called relative constant-weight codes, which
have applications to secret sharing schemes, can be obtained by
repeating coordinates.

Related to the Tanner graph itself, the operation of repeating
coordinates is also useful.  In~\cite{etv} the authors studied the
effect of repeating coordinates on the Tanner graph.  In particular,
\cite[Lem.~3, Prop.~4]{etv} shows that any linear cycle-free code
with rate $\le 0.5$ can be obtained from a linear cycle-free code
with rate $>0.5$ by repeating coordinates.

Assume that $G$ is a generator matrix of an $[n,k,d]$ code $\mcC$, and the value assignment $m(\cdot)$ is defined
from $G$ (as discussed in Section 2).  Then any codeword
$c\in\mcC$ may be written as $c = u G$ for some $u\in\GF(2)^k$; denote
by $u^{\perp}$ the set of points in $PG(k\!-\!1,2)$ that are
perpendicular to $u$ (according to the usual inner product), and let
\[ T_c = \{ p \in u^{\perp} \mid m(p) \ge 1\}
\quad \text{ and }\quad \overline{T_c} = u^{\perp} \setminus T_c
= \{ p \in u^{\perp} \mid m(p) = 0 \} \,. \]
Note that $T_c$ corresponds to those column indices $i$ of $G$ where
$c_i = 0$.  Then, we have:

\begin{thm}\label{thno3}
  Let $G$, $\mcC$, and $m(\cdot)$ be as above and let $c$ be any
  codeword of $\mcC$ with minimum weight $d$.
  \begin{enumerate}
  \item Define an $[n',k]$ code $\mcC'$ generated by the matrix $G'$
    obtained from $G$ by increasing the values of some of the points $p\in T_c$.
    Then, \[ \rho(\mcC') \;\le\; \rho(\mcC)+(n'-n) \]
    for the max-fractional pseudoweight and for the BEC pseudoweight.
  \item Define an $[n',k]$ code $\mcC'$ generated by the matrix $G'$
    obtained by adding the points of the set $\overline{T_c}$ to the
    columns of $G$.  If each added point $p\in\overline{T_c}$ repeats
    at least $\lceil(1-1/k)d\rceil$ times in the columns of $G'$, then
    \[ \rho(\mcC') \;\le\; \rho(\mcC) + (n'-n) \]
    for the max-fractional pseudoweight and for the BEC pseudoweight.
  \end{enumerate}
\end{thm}

\begin{proof}  The proof for the BEC pseudoweight is similar to that
for the max-fractional pseudoweight,  thus only the proof for the
max-fractional pseudoweight is given.

  To prove 1), up to code equivalence, we may assume that $T_c$ is the
  set of the first~$t$ columns (points) of $G$, where $1\le t\le n\!-\!1$,
  and $G'$ determines the value assignment $m'(\cdot)$.  From the
  assumption, it follows that $\mcC'$ is an $[n',k,d]$ code and
  $m'(p_i)\ge m(p_i)$ for any $1\le i\le t$.  Assume $m'(p_i) - m(p_i)
  = \theta_i$ for $1\le i\le t$ and let $H$ be the parity-check matrix
  with $\rho(\mcC)$ rows of $\mcC$.  Put the points $p_i$ for $1\le
  i\le t$ in order after the $n$-th column of $G$ to obtain $G'$, and
  each point~$p_i$, $1\leq i\leq t$, repeats $\theta_i$ times,
  respectively.  Construct a matrix $H'$ as
  \[ H'= \begin{pmatrix}
    H & & & \\
    * & \!H_1\! & & \\[-1.2ex]
    * & &\!\ddots\!& \\
    * & & & \!H_t~ \\
  \end{pmatrix}, \]
  where $H_i$ is an identity matrix of order $\theta_i$, and ``$*$''
  corresponding to $H_i$ has entries one at the $i$-th column of $H$
  for $1\le i\le t$.  It can be checked that $H'$ is a matrix with
  $\rho(\mcC) + \sum_{i=1}^t\theta_i = \rho(\mcC)+(n'-n)$ rows and
  of rank $(n-k) + \sum_{i=1}^t\theta_i = (n-k)+(n'-n) = n'-k$.
  Furthermore, $H'G'^T=0$, thus, $H'$ is a parity-check matrix of
  $\mcC'$.

  Assume now that $x\in\mcK(H')$.  Then we may write $x = (y, z)$,
  where $y = (y_1, \dots, y_n)\in \mcK(H)$ and $z = (z_1, \dots,
  z_{n'-n}) \in \R^{n'-n}$.  Note that $z$ may be rewritten as $z =
  (y_1,\dots,y_1, \,\cdots, y_t,\dots,y_t)$ where each $y_i$, $1\le
  i\le t$, repeats $\theta_i$ times, respectively.  Thus, for
  any $x\in\mcK(H')$, we have
  \begin{align*}
    w(x) &\,=\, \frac{\sum_{j=1}^n y_j + \sum_{i=1}^{n'-n} z_i}
    {\max_{j,i}\{y_j, z_i\}} \\
    &\,=\, \frac{\sum_{j=1}^n y_j + \sum_{i=1}^t \theta_i y_i}
    {\max_j \{y_j\}} \\
    &\,\ge\, \frac{\sum_{j=1}^n y_j}{\max_j\{y_j\}}
    \,\ge\, w_{\min}(H) \,=\, d \,.
  \end{align*}
  Thus, $w_{\min}(H')\ge d = d(\mcC')$, and so the result holds. \medskip


  For 2), it follows from the assumption that $\mcC'$ is an $[n',k,d]$
  code.  Let $G'$ be obtained by adding $t$ points $p_i$ in order,
  $1\leq i\leq t$, of the set $\overline{T_c}$ to the columns of $G$
  and assume that each point $p_i$ repeats $\theta_i\geq
  \lceil(1-1/k)d\rceil$ times, respectively.  Since $\mcC$ is a
  $k$-dimensional code, there exist basis points $b_1,\dots,b_k$ in
  the columns of $G$, and we may suppose without loss of generality
  that $b_j$ is in the $j$-th position in $G$, for $1\le j\le k$.
  Write each point $p_i$ as $p_i = \sum_{j=1}^k c_{ij} b_j$ and denote
  the support set by $A_i := \{ j \mid c_{ij} = 1 \} \subset
  \{1,\dots,k\}$, for $1\le i\le t$.

  Let $H$ be the parity-check matrix with $\rho(\mcC)$ rows of
  $\mcC$.  Construct a matrix $H'$ as
  \[ H' =
  \begin{pmatrix}
    H''\, \\
    h_1 \\
    \vdots \\
    h_t
  \end{pmatrix} , \]
  where
  \[ H'' =
  \begin{pmatrix}
    H & & & \\
    & \!H_1\! & & \\[-1.2ex]
    & & \!\ddots\! & \\
    & & & \!H_t~
  \end{pmatrix} ; \]
  here, each $H_i$, $1\le i\le t$, is a $(\theta_i \!-\! 1) \times
  \theta_i$ submatrix whose entries are defined as follows
  \begin{equation}\label{eq:subm}
    (H_i)_{st} =
    \begin{cases}
      1~ & \text{if }t\in\{s,s\!+\!1\},\\
      0~ & \text{otherwise,}
    \end{cases}
  \end{equation}
  and each $h_i$, $1\le i\le t$, is a binary vector with coordinate
  positions of $h_i$ equal to one whenever the position is in $A_i$ or
  the position corresponds to the first column of $H_i$.

  Since $\sum_{i=1}^t \theta_i = n'-n$, the matrix $H'$ is a
  $(\rho(\mcC)+(n'-n)) \times n'$ matrix of rank $n-k+(n'-n) = n'-k$.
  Furthermore, $H'G'^{T}=0$, thus, $H'$ is a parity-check matrix of
  $\mcC'$.

  Let $x\in \mcK(H')$.  Then, $x$ may be written as $x = (y, z)$,
  where $y = (y_1, \dots, y_n) \in \mcK(H)$, and $z\in \R^{n'-n}$.
  Note that $z$ may be written as $z = (z_1,\dots,z_1, \,\cdots,
  z_t,\dots,z_t)$, and each~$z_i$, $1\le i\le t$, repeats $\theta_i$
  times, respectively.

  If $\max_{j,i}\{y_j,z_i\} = \max_j\{y_j\}$, where $1\le j\le n$
  and $1\le i\le t$, then
  \begin{align*}
    w(x) &\,=\, \frac{\sum_{j=1}^ny_j + \sum_{i=1}^t\theta_iz_i}
    {\max_{j,i}\{y_j, z_i\}} \\
    &\,=\, \frac{\sum_{j=1}^ny_j + \sum_{i=1}^t\theta_iz_i}{\max_j\{y_j\}} \\
    &\,\ge\, \frac{\sum_{j=1}^n y_j}{\max_j\{y_j\}} \\
    &\,\ge\, w_{\min}(H) \,=\, d \,=\, d(\mcC')
    \qquad (\text{by } y \in \mcK(H)) \,.
  \end{align*}
  If $\max_{j,i}\{y_j, z_i\} = \max_i\{z_i\} = z_{i_0}$, then
  by the fundamental cone inequalities~(\ref{eq:cone}) we have
  $z_{i_0} \le \sum_{j\in A_{i_0}} y_j \le k \max_j \{y_j\}$,
  and since $\theta_{i_0} \ge (1-1/k)d$ we conclude
  \begin{align*}
    w(x) &\,=\, \frac{\sum_{j=1}^ny_j + \sum_{i=1}^t\theta_iz_i}
    {\max_{j,i}\{y_j, z_i\}} \\
    &\,=\, \frac{\sum_{j=1}^ny_j + \sum_{i=1}^t\theta_iz_i}{z_{i_0}} \\
    &\,\ge\, \frac{\sum_{j=1}^ny_j}{z_{i_0}} + \theta_{i_0} \\
    &\,\ge\, \frac{\sum_{j=1}^ny_j}{k\max_j\{y_j\}} + \theta_{i_0} \\
    &\,\ge\, (1/k)d+(1-1/k)d \qquad (\text{by } y \in \mcK(H)) \\
    &\,=\, d \,=\, d(\mcC') \,.
  \end{align*}
  Thus, $w_{\min}(H')\geq d(\mcC')$ in any case, and so the
  result follows.
\end{proof}

\begin{rem}
 Whether the above theorem holds for the other two  pseudoweights is
 an open problem.
\end{rem}

Let $\mcC$ be a binary linear code of length $n$ and let $\mcI'
\subset \mcI = \{1,2,\dots, n\}$ be a subset of $\mcI$.  Define \[
\mcC_{\mcI'} \,=\, \{c\in\mcC \mid \text{supp}(c) \subset \mcI'\}
\,, \] which is the {\em shortened subcode} of $\mcC$ supported by
$\mcI'$. Regarding the pseudoredundancy of the shortened subcode
${\mcC}_{{\mcI}'}$, the result is as follows.

\begin{thm}\label{thno4}
  Let $\mcC$ be an $[n,k,d]$ code and let $c \in \mcC$ be any codeword
  of minimum weight~$d$.  Then, for any shortened subcode
  ${\mcC}_{{\mcI}'}$ containing the codeword $c$, we have
  \[ \rho({\mcC}_{{\mcI}'}) \;\le\; \rho(\mcC) + (n-|{\mcI}'|) \]
  for all the four pseudoweights.
\end{thm}

\begin{proof}
  Let $H$ be a parity-check matrix of $\mcC$ with $\rho(\mcC)$ rows
  and let $H_{\mcI'}$ be the submatrix of~$H$ consisting of the
  columns corresponding to $\mcI'$.  Define $\mcC'_{\mcI'}$ as the
  code obtained by puncturing those columns of $\mcC_{\mcI'}$
  corresponding to $\mcI\setminus{\mcI'}$.  Then it can be checked
  that $H_{\mcI'}$ is a parity-check matrix of $\mcC'_{\mcI'}$.  Since
  $\mcC'_{\mcI'}$ is a linear code with minimum distance $d$ according
  to the assumption, and since for any $x\in \mcK(H_{\mcI'})$ and
  $(x,0)\in \R^n$ we have $(x,0)\in \mcK(H)$, it follows that
  \[ w_{\min}(H_{\mcI'}) \,\ge\, w_{\min}(H)
  \,=\, d(\mcC) \,=\, d \,=\, d(\mcC'_{\mcI'}) \,. \]

  Thus, $\rho(\mcC'_{\mcI'})\leq \rho(\mcC)$.  Then, using the proof
  of Lemma~4.1 in~\cite{zsf}, we get $\rho(\mcC_{\mcI'})\le
  \rho(\mcC'_{\mcI'}) + (n-|\mcI'|)\le \rho(\mcC)+(n-|\mcI'|)$.
\end{proof}

A code $\mcC$ is called {\em subcode-complete} if any subcode $\mcD$
of $\mcC$ can be written as $\mcD = \mcC_{\mcI'}$ for some
$\mcI'\subset\mcI$.  Define $\supp(\mcD) =
\bigcup_{c\in\mcD}\supp(c)$.  Since $\supp(\mcD) = \bigcap \{\mcI'\mid
\mcC_{\mcI'}\supset\mcD\}$, it follows that a code $\mcC$ is
subcode-complete if and only if $\mcD = \mcC_{\supp(\mcD)}$ for any
subcode $\mcD$ of $\mcC$.  The following result gives a judging rule
for a code to be subcode-complete by using the value assignment.

\begin{thm}\label{thno5}
  A code $\mcC$ with value assignment $m(\cdot)$ is subcode-complete
  if and only if ${m(p) > 0}$ for all $p\in PG({k\!-\!1},2)$.
\end{thm}

\begin{proof}
  Let $\mcC$ be subcode-complete.  Assume that $G$ is a generator
  matrix corresponding to $m(\cdot)$.  If there exists a point $p_0$
  such that $m(p_0) = 0$, then consider the $(k\!-\!1)$-dimensional
  subspace $(p_0)^{\perp}$ of $\GF(2)^k$, where
  \[ (p_0)^{\perp} \,=\, \{v\mid  v \text{ is perpendicular to } p_0\} \,. \]
  It follows that $\mcD=\{c\mid c=vG \text{ and } v\in
  (p_0)^{\perp}\}$ is a $(k\!-\!1)$-dimensional subcode of $\mcC$.
  Since $m(p_0)=0$, we get $\supp(\mcD)=\supp(\mcC)$.  Thus, $\mcD\neq
  \mcC_{\supp(\mcD)}=\mcC_{\supp(\mcC)}=\mcC$, a contradiction to that
  $\mcC$ is subcode-complete.

  Conversely, suppose that $m(p)>0$ for each $p\in PG(k\!-\!1,2)$ and
  consider any $r$-dimensional ($1\le r\le k$) subcode $\mcD$.  Note
  that a generator matrix of $\mcD$ can be written as $U_{r\times k}G$
  for some matrix $U_{r\times k}$.  Denote
  \[ (U_{r\times k})^{\perp} \,=\, \{p\mid p\in PG(k\!-\!1,2)
  \text{ and $p$ is perpendicular to each row of } U_{r\times k}\} \,. \]
  Then, $(U_{r\times k})^{\perp}$ is a $(k\!-\!r\!-\!1)$-dimensional
  subspace of $PG(k\!-\!1,2)$.

  Since $m(p)>0$ for each $p\in PG(k\!-\!1,2)$,  we get that the set
  \[ W \,=\, \{p\mid p\in (U_{r\times k})^{\perp}\text{ and }m(p)>0\} \]
  is equal to the $(k\!-\!r\!-\!1)$-dimensional projective subspace
  $(U_{r\times k})^{\perp}$.  Observe that $\supp(\mcD)$ corresponds to
  those columns of $G$ (considered as points of $PG(k\!-\!1,2)$) which are
  not contained in $W = (U_{r\times k})^{\perp}$.  Thus,
  \begin{align*}
    \mcC_{\supp(\mcD)} &\,=\,
    \{ c\mid c = vG \text{ and $v\in\GF(2)^k$ is} \\
    &\qquad\qquad\text{perpendicular to each point in }
    W = (U_{r\times k})^{\perp} \} \\
    &\,=\, \{ c\mid c = vG \text{ and $v\in GF(2)^k$ is} \\
    &\qquad\qquad\text{a linear combination of the rows of }
    U_{r\times k} \} \\
    &\,=\, \mcD \,.
  \end{align*}
  Thus, $\mcC$ is subcode-complete.
\end{proof}

From Theorems~\ref{thno4} and~\ref{thno5}, one gets the following
result.

\begin{cor}
  Let $\mcC$ be a subcode-complete $[n,k,d]$ code and let $c$ be any
  codeword with minimum weight~$d$.  Then, for any subcode $\mcD$
  containing $c$, there holds
  \[ \rho(\mcD) \;\le\; \rho(\mcC) + (n-|\supp(\mcD)|) \]
  for all the four pseudoweights.
\end{cor}

We may show that some special subcode-complete codes have finite
pseudoredundancy and one example of such codes is a binary linear
constant-weight code.
Since any binary linear constant-weight code consists of copies of a
binary simplex code, or equivalently,  the value assignment of a
linear constant-weight code  takes the same value at each point
$p\in
  PG(k\!-\!1,2)$,   a linear constant-weight code is subcode-complete by Theorem~\ref{thno5}.
 In~\cite{zsf} it is shown that a binary simplex code has finite
pseudoredundancy as follows.

\begin{lem}[{\cite[Prop.~7.8]{zsf}}]\label{simplex}
  For $k\geq 2$, the $[2^k\!-\!1, k, 2^{k-1}]$ simplex code $\mcC$
  satisfies
  \[ \rho(\mcC) \;\le\; \frac{(2^k-1)(2^{k-1}-1)}{3} \]
  for all the four pseudoweights.
\end{lem}

In the proof of Lemma~\ref{simplex} (see~\cite{zsf}), a parity-check
matrix $H'$ of $\mcC$ is chosen such that the rows of $H'$ consist of
all the codewords of the Hamming code (the dual code of the simplex
code $\mcC$) with Hamming weight equal to 3.  In our viewpoint, the
value assignment of the simplex code~$\mcC$ satisfies $m(p) = 1$ for
any $p\in PG(k\!-\!1,2)$, that is, the columns of a generator matrix
of~$\mcC$ are exactly all the different points in $PG(k\!-\!1,2)$.  By
using such a viewpoint, we may give an alternative explanation of the
bound in $\rho(\mcC)$ in Lemma~\ref{simplex}.  Since any row of $H'$
can be viewed as a linear relation of three different columns of the
generator matrix of $\mcC$, any row of $H'$ can also be viewed as a
{\em line} (spanned by two projective points) in $PG(k\!-\!1,2)$.
Thus, the number of rows in $H'$ equals the number of lines in
$PG(k\!-\!1,2)$, which is $\frac 13 (2^k-1) (2^{k-1}-1)$.

By using Lemma \ref{simplex} and the structure of a linear
constant-weight code, we obtain:

\begin{thm}
  Any binary linear $[n,k,d]$ constant-weight code $\mcC$  satisfies
  \[ \rho(\mcC) \;\le\; n + \frac{(2^k-1)(2^{k-1}-4)}{3} \]
  for all the four pseudoweights.
\end{thm}

\begin{proof}
  Assume the value assignment of the given binary $[n,k,d]$
  constant-weight code is $m(\cdot)$.   Then,  $m(\cdot)$ takes the same value at each point $p\in PG(k\!-\!1,2)$,  and then  one may
  get that $n=(2^k-1)m(p)$ and $d=2^{k-1}m(p)$ for any point $p\in
  PG(k\!-\!1,2)$.

  Arrange a generator matrix $G$ of the constant-weight code as
  follows: put each point $p\in PG(k\!-\!1,2)$ in some fixed order
  once in the columns of $G$, and then, in the same order as before,
  repeat each of these points $m(p)-1$ times in the columns of $G$.
  According to such a matrix~$G$, a parity-check matrix $H$ of $\mcC$
  can be constructed as follows:
  \[ H \,=\, \begin{pmatrix}
    ~H' & {\bf 0}~ \\
    ~* & I~ \\
  \end{pmatrix} , \]
  where $H'$ is the parity-check matrix of the simplex code given
  below Lemma \ref{simplex}, $\bf 0$ stands for a ${\frac 13
    (2^k\!-\!1) (2^{k-1}\!-\!1)} \times {(m(p)\!-\!1)(2^k\!-\!1)}$
  zero matrix, and $I$ stands for a ${(m(p)\!-\!1)(2^k\!-\!1)} \times
  {(m(p)\!-\!1)(2^k\!-\!1)}$ identity matrix; finally, $*$ stands for
  a ${(m(p)\!-\!1)(2^k\!-\!1)} \times {\frac 13
    (2^k\!-\!1)(2^{k-1}\!-\!1)}$ matrix, which is written as
  \[ * \,=\,
  \begin{pmatrix}
    H_1 \\
    \vdots \\
    ~H_{2^k-1}~
  \end{pmatrix} , \]
  where each~$H_i$, $1 \le i \le 2^k\!-\!1$, is an ${(m(p)\!-\!1)}
  \times {\frac 13 (2^k\!-\!1)(2^{k-1}\!-\!1)}$ matrix, which has
  entries zero except for its $i$-th column, whose entries are all
  equal to one.

  It can be checked that $H$ is a matrix satisfying $HG^T = 0$ and
  \begin{align*}
    \rank(H) &\,=\, {\rank(H') + (m(p)\!-\!1)(2^k\!-\!1)} \\
    &\,=\, (2^k\!-\!1) - k + (m(p)\!-\!1)(2^k\!-\!1) \\
    &\,=\, m(p)(2^k\!-\!1) - k \,=\, n - k \,.
  \end{align*}
  Thus, $H$ is a parity-check matrix of $\mcC$.

  For this parity-check matrix $H$, let $x\in \mcK(H)$. Then,
  according to the fundamental cone inequalities (\ref{eq:cone}), $x$
  may be written as $x = (y, z)$, where $y = (y_1, \dots, y_{2^k-1})
  \in \mcK(H')$, and $z=(z_1, \dots, z_{n-2^{k}+1})\in \R^{n-2^k+1}$
  is obtained from $y$ by repeating $(m(p)\!-\!1)$ times each
  coordinate of $y$.  Thus, $w_{\text{maxfrac}}(x)$ can be computed as follows.
  \begin{align*}
    w_{\text{maxfrac}}(x) &\,=\, \frac{\sum_{i=1}^{2^k-1}y_i + \sum_{j=1}^{n-2^k+1}z_j}
    {\max_{i,j}\{y_i, z_j\}} \\
    &\,=\, \frac{\sum_{i=1}^{2^k-1}y_i + (m(p)-1)\sum_{i=1}^{2^k-1}y_i}
    {\max_{i}\{y_i\}} \\
    &\,=\, m(p) \frac{\sum_{i=1}^{2^k-1}y_i}{\max_{i}\{y_i\}} \\
    &\,=\, m(p) w_{\text{maxfrac}}(y) \\
    &\,\ge\, m(p) 2^{k-1} = d \quad (\text{by } y\in\mcK(H')
    \text{ and Lemma~\ref{simplex}}).
  \end{align*}
  Thus, $w_{\min}(H) \ge d$ for the max-fractional pseudoweight,  and thus
  the result of the theorem holds for the max-fractional pseudoweight   by the fact that  $H$ is a matrix with $\frac 13
  (2^k\!-\!1) (2^{k-1}\!-\!1) + (m(p)\!-\!1)(2^k\!-\!1) = n +
  \frac 13 (2^k\!-\!1) (2^{k-1}\!-\!4)$ rows.   It follows that  the result
  holds  also  for all the  four  pseudoweights  by  (\ref{eq:inequality}).
\end{proof}


\section{$k$-dimensional Codes Constructed by
  Value Assignment}\label{sec:k-dim}

In this section, we will proceed to determine the pseudocodeword
redundancies of certain $k$-dimensional binary codes by making use of
the value assignment.

Let $\mcC$ be an $[n,k]$ binary code determined by a value assignment
$m(\cdot)$.  Then, a basic fact is that
\[ \!\sum_{p\in PG(k\!-\!1,2)}\! m(p) \,=\, n \,. \]

We will use and extend the following results.

\begin{lem}[{\cite[Lem.~6.1]{zsf}}]\label{lem:w2}
  Let $H$ be a parity-check matrix of $\mcC$ such that every row
  in~$H$ has weight 2.  Then:
  \begin{enumerate}
  \item There is an equivalence relation on the set~$\mcI$ of column
    indices of~$H$ such that for a vector $x\in\R^n$ with nonnegative
    coordinates, we have $x\in \mcK(H)$ if and only if $x$ has
    equal coordinates within each equivalence class.
  \item The minimum distance of $\mcC$ is equal to its minimum BEC,
    AWGNC, BSC, and max-fractional pseudoweights with respect to $H$,
    i.e., $d(\mcC)=w_{\min}(H)$.
  \end{enumerate}
\end{lem}

\begin{lem}[{\cite[Prop.~6.2]{zsf}}]\label{lem:w2+}
  Let $H$ be an $m\times n$ parity-check matrix of $\mcC$, and assume
  that the $m\!-\!1$ first rows in $H$ have weight 2.  Denote by
  $\widetilde{H}$ the $(m\!-\!1)\times n$ matrix consisting of these rows,
  and consider the equivalence relation of Lemma~\ref{lem:w2}-2.\ with
  respect to $\widetilde{H}$, and assume that~$\mcI_m$ intersects each
  equivalence class in at most one element.  Then, the minimum
  distance of $\mcC$ is equal to its minimum BEC, AWGNC, BSC, and
  max-fractional pseudoweights with respect to $H$, i.e.,
  $d(\mcC)=w_{\min}(H)$.
\end{lem}

Using these lemmas, in~\cite[Cor.~6.4]{zsf} it was shown that all
2-dimensional binary codes $\mcC$ with length $n$ have pseudoredundancy
$\rho(\mcC)=n-2$, and the proof was conducted according to the
analysis of the {\em supports} of the two codewords generating the
$2$-dimensional code.

By the viewpoint of the value assignment, we may consider the
different cases of the supports of the two codewords generating the
$2$-dimensional code as {\em different points} in $PG(1, 2)$.
Generalizing this idea, one may consider for each point occurring in
the columns of a generator matrix of~$\mcC$ the corresponding
equivalence class from Lemma~\ref{lem:w2}, and the size of this
equivalence class is exactly the value of the corresponding value
assignment at this point.  In such a viewpoint, $\mcI_m$ in the
parity-check matrix in Lemma~\ref{lem:w2+} is exactly the linear
relation among different points in the columns of the generator matrix
of $\mcC$.

Using the above stated techniques and Lemmas~\ref{lem:w2}
and~\ref{lem:w2+}, we will in this section construct several kinds of
$[n,k]$ codes whose pseudoredundancies are equal to $n-k$.  The first
result is:

\begin{thm}\hfill
  \begin{enumerate}
  \item For any $k$ independent points $p_i\in PG(k\!-\!1,2)$, $1\le
    i\le k$, if a value assignment $m(\cdot)$ satisfies
    \[ m(p) = \begin{cases}
      z_i\ge 1~ & \text{if }p=p_i,\ 1\le i\le k\,,\\
      0 & \text{otherwise},
    \end{cases} \]
    and there exists some $1\le i_0\le k$ such that $m(p_{i_0}) =
    z_{i_0}\ge 2$, then the $[n,k]$ code $\mcC$ determined by
    $m(\cdot)$ satisfies
    \[ \rho(\mcC) = \textstyle\sum\limits_{i=1}^k z_i-k = n-k \]
    for all the four pseudoweights.

  \item For any $k+1$ points $p_i\in PG(k\!-\!1,2)$, $1\le i\le k+1$,
    such that the points~$p_i$, $1\leq i\leq k$, are independent, if a
    value assignment satisfies
    \[ m(p) = \begin{cases}
      z_i\ge 1~ &  \text{if }p=p_i,\ 1\le i\le  k+1\,,\\
      0 & \text{otherwise},
    \end{cases} \]
    then the $[n,k]$ code $\mcC$ determined by $m(\cdot)$ satisfies
    \[ \rho(\mcC) = \textstyle\sum\limits_{i=1}^{k+1}z_i-k = n-k \]
    for all the four pseudoweights.
  \end{enumerate}
\end{thm}

\begin{proof}
  For 1), up to code equivalence, we may arrange a generator
  matrix~$G$ of~$\mcC$ in such a way that the first $m(p_1)$ columns
  of $G$ are the point~$p_1$, the next $m(p_2)$ columns of $G$ are the
  point~$p_2$, and in such an order, one proceeds to put the
  point~$p_k$ in the last $m(p_k)$ columns of~$G$.  For this
  matrix~$G$, we construct a matrix $H$ in block diagonal form
  \[ H = \begin{pmatrix}
    H_1\quad\quad~~ \\[-1.2ex]
    \quad~ \ddots \quad~ \\[-.5ex]
    \quad\quad~~ H_k
  \end{pmatrix}, \]
  where $H_i$ is an $(m(p_i)\!-\!1)\times m(p_i)$ submatrix whose
  entries are defined as in (\ref{eq:subm}).  It can be checked that
  $H$ is an $(n-k)\times n$ matrix of rank $n-k$ and $HG^T=0$.  Thus,
  $H$ is a parity-check matrix of $\mcC$, and so $\rho(\mcC) = n-k$ by
  Lemma~\ref{lem:w2}.

  For 2), since the points $p_i$ for $1\le i\le k$ are a basis for
  $PG(k\!-\!1,2)$, one may write $p_{k+1}$ as a linear combination of
  these basis points.  Up to code equivalence, one may write $p_{k+1}
  = \sum_{j=1}^s p_j$ for $s\le k$.  Arrange a generator matrix $G$
  of~$\mcC$ similarly to the proof of 1), that is, the first $m(p_1)$
  columns are the point $p_1$, the next $m(p_2)$ columns are the
  point $p_2$, and in such an order, the last $m(p_{k+1})$ columns are
  the point $p_{k+1}$.  Then, we may construct a matrix $H$ as
  \[ H = \begin{pmatrix}
    \,H'\, \\
    h
  \end{pmatrix}, \]
  where the submatrix  $H'$ is the block diagonal one
  \[ H' = \begin{pmatrix}
     H_1 \quad\quad\quad~~~~~~ \\[-1ex]
    \quad~ \ddots \quad\quad~~~ \\[-.3ex]
    \quad\quad~~ H_k \quad~~ \\
    \quad\quad\quad~~~~~ H_{k+1}\, \\
  \end{pmatrix}, \]
  and $H_i$ for $1\le i\le k+1$ is an $(m(p_i)\!-\!1)\times m(p_i)$
  matrix defined as in (\ref{eq:subm}), and $h$ is a binary row vector
  whose coordinate positions corresponding to the first column of each
  $H_i$ for $1\le i\le t$ and to the first column of $H_{k+1}$ are
  equal to one.  It can be checked that $H$ is an $(n-k)\times n$
  matrix of rank $n-k$ and $HG^T=0$, and so $H$ is a
  parity-check matrix of $\mcC$. Thus, $\rho(\mcC) = n-k$ by
  Lemma~\ref{lem:w2+}.
\end{proof}

In order to determine the pseudocodeword redundancies of more kinds
of $k$-dimensional codes, it is convenient to introduce the
following notations.  Let $p_1, \dots, p_k$ be the points of a basis
of $PG(k\!-\!1,2)$.  Then, any $p\in PG(k\!-\!1,2)$ may be written
as $p=\sum_{i=1}^kc_ip_i$, where $c_i\in \GF(2)$ for $1\le i\le k$.
Call the set of the basis points $p_i$ whose coefficients are
nonzero the {\em representing-set} of the point $p$ with respect to
the basis points $p_1,\dots,p_k$.  If the basis points are clear,
one may simply call this set representing-set of the point $p$.

In the following text of this section, for basis points $p_1, \dots,
p_k$ of $PG(k\!-\!1,2)$, let $S_1,\dots,S_t$ stand for the
representing-sets of $p_{k+1},\dots,p_{k+t}$, respectively.

\begin{defn}
  The points $p_{k+1},\dots,p_{k+t}$ are called
  {\em representing-independent}
  if their representing-sets $S_1,\dots,S_t$ are pairwise disjoint.
  They are called {\em representing-dependent} if
  for all $1\le i\le t$ there exists $1\le j\le t$, $j\ne i$
  such that $S_i\cap S_j \ne \varnothing$.
\end{defn}

For the points that are representing-independent, we have:

\begin{thm}\label{thm:3}
  For any $k+t$ points $p_i\in PG(k\!-\!1,2)$, $1\le i\le k+t$, such that
  the points $p_1,\dots,p_k$ are basis points and the points
  $p_{k+1},\dots,p_{k+t}$ are representing-independent, if the value
  assignment $m(\cdot)$ satisfies
  \[ m(p) = \begin{cases}
    z_i\ge 1~ & \text{if }p=p_i,\ 1\le i\le k+t\,,\\
    0 & \text{otherwise},
  \end{cases} \]
  then the $[n,k]$ code $\mcC$ determined by $m(\cdot)$ satisfies
  $\rho(\mcC)=\sum_{i=1}^{k+t}z_i-k=n-k$ for all the four pseudoweights.
\end{thm}

\begin{proof}
  Still arrange a generator matrix $G$ of $\mcC$ in such a way that
  the first $m(p_1)$ columns of $G$ are the point~$p_1$ and the last
  $m(p_{k+t})$ columns are the point~$p_{k+t}$.  Furthermore, up to
  code equivalence, we may assume that the representing-set of
  $p_{k+j}$ is $S_j = \{p_{s_{j-1}+1}, p_{s_{j-1}+2}, \dots, p_{s_j}\}$, $1\le
  j\le t$, where $0 = s_0 < s_1 < s_2 < \cdots < s_t \le k$.
  Construct a matrix $H$ as
  \begin{equation}\label{eq:ori}
    H = \begin{pmatrix}
      \,H'\, \\
      h_1 \\
      \vdots \\
      h_t
    \end{pmatrix},
  \end{equation}
  where
  \[ H'= \begin{pmatrix}
    H_1 \quad\quad\quad\quad\quad~~~~~~~~ \\
    \quad~ \ddots \quad\quad\quad\quad~~~~~~~ \\
    \quad\quad~~ H_k \quad\quad\quad~~~~~~ \\
    \quad\quad\quad~~~~ H_{k+1} \quad\quad~~~~ \\
    \quad\quad\quad\quad~~~~~~ \ddots \quad~~ \\
    \quad\quad\quad\quad\quad~~~~~~~~ H_{k+t}\, \\
  \end{pmatrix} \]
  is a block diagonal submatrix, and $H_i$ for $1\le i\le k+t$ is
  defined as in (\ref{eq:subm}), and $h_i$ for $1\le i\le t$ is a binary
  row vector, and we demand that the coordinate position of $h_j$,
  $1\le j\le t$, corresponding to the first column of $H_i$ for
  $s_{j-1}+1\le i\le s_j$ and to the first column of $H_{k+j}$ be
  equal to one.  Then, it can be checked that $HG^T=0$ and
  $\text{rank}(H)=n-k$, thus, $H$ is a parity-check matrix of $\mcC$.

  Consider the Tanner graph of the code $\mcC$ with respect to the
  parity-check matrix $H$.  It is easy to see that this Tanner graph
  is a disjoint union of trees, i.e., it does not have any cycles.
  From~\cite[Lem.~28]{vk} it follows that the fundamental polytope
  equals the code polytope.  Therefore, there do not exist any proper
  pseudocodewords.  Hence, it holds that $\rho(\mcC) = n-k$ for all
  the four pseudoweights.
\end{proof}

For the representing-dependent case, it is more complicated to
determine the pseudocodeword redundancy, as the codes will have in
general no cycle-free Tanner graph representation.  However, we may
get some results about some particular codes.

Assume that $p_1,\dots,p_k$ are basis points of $PG(k\!-\!1,2)$ and
that $p_{k+1},\dots,p_{k+t}$ are representing-dependent.  Denote by
$U_1$ the basis points which belong to one and only one representing-set
and denote by $U_2$ the basis points which belong to at least two
representing-sets, so that $U_2 = (\bigcup_{i=1}^t S_i)\setminus U_1$.
Define $U_3 = U_1\cup \{p_{k+1}, p_{k+2}, \dots, p_{k+t}\}$.

\begin{thm}\label{thm:4}
  Let the notations be defined as above and assume one of the
  following conditions holds:
  \begin{enumerate}
  \item $S_i\cap U_1\ne \varnothing$ for each $1\le i\le t$, and
    $\min\{m(p) \mid p\in U_2\} \ge \max\{m(p)\mid p\in U_3\}$,
  \item $W_1 = \{j\mid S_j\cap U_1 \ne \varnothing\} \ne \varnothing$ and
    $W_2 = \{j\mid S_j\cap U_1 = \varnothing\} \ne \varnothing$ and
    $\min\{m(p) \mid p\in U_2\} \ge \max\{m(p)\mid p\in U_3\}$ and
    $\min\{m(p_{k+j})\mid j\in W_1\} \le \min\{m(p_{k+j})\mid j\in W_2\}$,
  \item $|\bigcap_{i=1}^{t}S_i| \ge 2$ and $\max\{m(p)\mid p\in
    \bigcap_{i=1}^{t}S_i\} \le \min \{m(p)\mid p\in
    ((\bigcup_{i=1}^{t}S_i) \setminus (\bigcap_{i=1}^{t}S_i))\cup\{p_{k+1},
    \cdots, p_{k+t}\}\}$.
  \end{enumerate}
  Then, the $[n,k]$ code $\mcC$ determined by
  \[ m(p) = \begin{cases}
    z_i\ge 1~ &  \text{if }p=p_i,\ 1\le i\le k+t\,,\\
    0  & \text{otherwise},
  \end{cases} \]
  satisfies $\rho(\mcC)=n-k$ for all the four pseudoweights.
\end{thm}

\begin{proof}
  The proof is similar for the three cases.  We only give the proof
  for the first case.  Arrange a generator matrix $G$ of $\mcC$ as
  before, namely, put the points $p_i$ for $1\le i\le k+t$ in order in
  the columns of $G$, and each point $p_i$, $1\le i\le k+t$, repeats
  $m(p_i)$ times.

  Construct a matrix $H$ as in (\ref{eq:ori}), and the binary vector
  $h_j$ in $H$, $1\le j\le t$, is determined by the representing-set
  $S_j$ of the point $p_{k+j}$.  If $S_j = \{p_{i_1}, p_{i_2}, \dots,
  p_{i_{\theta}}\}$, then $h_j$ has a one in each coordinate position
  corresponding to the first column of $H_{k+j}$ and to the first
  column of $H_{i_l}$, $1\le \ell\le \theta$.  Then it can be checked
  that $HG^T=0$ and $\text{rank}(H) = \sum_{i=1}^{k+t}m(p_i)-k =
  n-k$. Thus, $H$ is a parity-check matrix of $\mcC$.

  For $1\le i\le t$, let $T_i=S_i\cup \{p_{k+i}\}$, and let $V = \{p_1,
  \dots, p_k\} \setminus (\bigcup_{i=1}^tS_i)$.  Define
  \[ \begin{split}
    \gamma_i &= \min\{m(p_{j_1})+m(p_{j_2})\mid
    p_{j_1}, p_{j_2}\in T_i\setminus U_2 ,\, j_1\ne j_2\} \\
    \gamma &= \min_{1\le i\le t}\{\gamma_i\} \\
    \delta &= \min\{m(p)\mid p\in V\} \,.
  \end{split} \]

  Different from the representing-independent case, the
  analysis of the codewords with minimum (Hamming) weight is tedious
  in the representing-dependent case.  In general, according
  to the construction of the parity-check matrix $H$, one may divide
  the possible codewords with minimum weight into two classes: one
  class is the codewords with nonzero coordinate in the position
  corresponding to some point in $U_2$, and the other class is the
  ones with zero coordinate in the position corresponding to any point
  in $U_2$ (note that $U_2 = \varnothing$ in the
  representing-independent case).  The confined condition
  $\min\{m(p) \mid p\in U_2\}\ge \max\{m(p)\mid p\in U_3\}$ plays a
  key role in determining the codewords with minimum weight.

  More concretely, since $\min\{m(p) \mid p\in U_2\}\ge \max\{m(p)\mid
  p\in U_3\}$, it follows that $d(\mcC)=\min\{\gamma, \delta\}$ by
  analyzing the constructed parity-check matrix $H$, that is, the
  codewords with minimum weight should be ones with zero coordinates
  in the positions corresponding to any point in $U_2$.

  On the other hand, for $x\in \mcK(H)$ and $p\in T_i$, $1\le
  i\le t$, we have
  \[ x_p \le \sum_{p'\in T_i\setminus\{p\}}x_{p'} \qquad
  \text{for all }p\in T_i \,, \]
  and thus by the assumption $\min\{m(p) \mid p\in
  U_2\}\geq \max\{m(p)\mid p\in U_3\}$, we get
  \begin{align*}
    m(p'') x_{p} &\le m(p'')
    \big( \! \sum_{p'\in T_i\setminus\{p\}} \! x_{p'} \big)
    \quad (\text{for some }p'' \in T_i\setminus U_2) \\
    & \le \! \sum_{p'\in T_i\setminus\{p\}} \! m(p') x_{p'} \,.
  \end{align*}
  Thus,
  \[ (m(p'')+m(p))x_p \,\le\, \sum_{p'\in T_i}m(p')x_{p'} \,, \]
  that is,
  \begin{equation}\label{eq:4}
    m(p'')+m(p) \,\le\, \frac{\sum_{p'\in T_i}m(p')x_{p'}}{x_p} \,.
  \end{equation}

  Since $\min\{m(p) \mid p\in U_2\}\ge \max\{m(p)\mid p\in U_3\}$,
  $\gamma_i\le m(p'')+m(p)$ always holds no matter $p\in (T_i\setminus
  U_2)$ or $p\in (T_i\cap U_2)$, $1\le i\le t$.  Thus, (\ref{eq:4})
  can be rewritten as
  \begin{equation}\label{eq:5}
    \gamma_i \,\le\, m(p'')+m(p) \,\le\,
    \frac{\sum_{p'\in T_i}m(p')x_{p'}}{x_p} \,.
  \end{equation}

  Therefore,
  \[ \begin{split}
    w_{\text{maxfrac}}(x) &\,=\, \frac{\sum_{j=1}^n x_j}{\max_{j}\{x_j\}} \\
    &\,=\, \frac{\sum_p m(p)x_p}{\max\{x_p\}} \,, \quad
    p\in \textstyle\bigcup\limits_{i=1}^tT_i \,\cup V \\
    &\,\ge\, \begin{cases}
      \gamma_i~ & \text{if }\max\{x_p\} = x_{p_0} \text{ and } p_0\in T_i,\
      1\le i\le t~\ (\text{by~\ref{eq:5}}), \\
      \delta & \text{if }\max\{x_p\} = x_{p_0} \text{ and }  p_0\in V \,.
    \end{cases}
  \end{split} \]
  Thus, $w_{\min}(H)\geq d(\mcC)=\min\{\gamma, \delta\}$ for the
  max-fractional pseudoweight.  It follows that  $\rho_{\text{maxfrac}}(\mcC)\leq
  n-k$.   We thus  have  $\rho(\mcC)\leq n-k$ for
all the four pseudoweights by  (\ref{eq:inequality}).   Since
$\rho(\mcC)\geq n-k$  for all the four pseudoweights   by Definition
\ref{def:pr},  we  have $\rho(\mcC)= n-k$ for all the four
pseudoweights.

\end{proof}

\begin{exa}\label{example}
  Let $k=6$ and $t=3$; consider the linear code~$\mcC$ generated by
  the matrix \[ G =
  \begin{pmatrix}
    ~1 & 0 & 0 & 0 & 0 & 0 & 1 & 0 & 0~ \\
    ~0 & 1 & 0 & 0 & 0 & 0 & 0 & 1 & 0~ \\
    ~0 & 0 & 1 & 0 & 0 & 0 & 0 & 0 & 1~ \\
    ~0 & 0 & 0 & 1 & 0 & 0 & 1 & 0 & 1~ \\
    ~0 & 0 & 0 & 0 & 1 & 0 & 1 & 1 & 0~ \\
    ~0 & 0 & 0 & 0 & 0 & 1 & 0 & 1 & 1~ \\
  \end{pmatrix}. \]

  Denote by $i$-th column of the matrix $G$ by $p_i$, $1\le i\le 9$.
  Then, $p_i$ for $1\le i\le 6$ are basis points of $PG(5, 2)$, and
  the points $p_7$, $p_8$, and $p_9$ have representing-sets $S_1 =
  \{p_1,p_4,p_5\}$, $S_2 = \{p_2,p_5,p_6\}$, and $S_3 =
  \{p_3,p_4,p_6\}$, respectively.  One sees that $U_1 = \{p_1, p_2,
  p_3\}$, $U_2 = (S_1\cup S_2\cup S_3) \setminus U_1 =
  \{p_4,p_5,p_6\}$, and $U_3 = U_1\cup\{p_7,p_8,p_9\} =
  \{p_1,p_2,p_3,p_7,p_8,p_9\}$.  Furthermore, $S_1\cap S_2\cap S_3 =
  \varnothing$ and $S_i\cap U_1 \ne \varnothing$ for each $1\le i\le 3$
  and $m(p_i)=1$ for $1\le i\le 9$.  Thus, the points $p_i$, $1\leq
  i\le 9$, satisfy and only satisfy the conditions of Case 1) in
  Theorem~\ref{thm:4}, and therefore, $\rho(\mcC) = n-k = 9-6 = 3$.

  We remark that the code $\mcC$ is not {\em cycle-free}, i.e., the
  Tanner graph of any parity-check matrix of $\mcC$ has a cycle, as we
  will now demonstrate.  According to the proof of
  Theorem~\ref{thm:4}, one may take a parity-check matrix $H$ of
  $\mcC$ as
  \[ H =\begin{pmatrix}
    \,1\ 0\ 0\ 1\ 1\ 0\ 1\ 0\ 0\, \\
    \,0\ 1\ 0\ 0\ 1\ 1\ 0\ 1\ 0\, \\
    \,0\ 0\ 1\ 1\ 0\ 1\ 0\ 0\ 1\, \\
  \end{pmatrix}. \]
  Obviously, the three rows of this $H$ are independent and $H$ has a
  cycle located at the 3-th, 4-th, and 5-th coordinates.

  Due to the fact that any parity-check matrix with 3 independent rows
  may be written as $PH$, where $P$ is a $3\times 3$ invertible
  matrix, it suffices to check that any matrix $PH$ has a cycle for
  any invertible $3\times 3$ matrix $P$.  In fact, one may list all
  binary invertible $3\times 3$ matrices $P$, and then check that $PH$
  has a cycle for each such $P$.  A simpler argument is to make use of
  the form of the matrix $H$. Observe that $H$ may be written as $(I,
  M, I)$, where $I$ is the $3\times 3$ identity matrix, and
  \[ M = \begin{pmatrix}
    \,1\ 1\ 0\, \\
    \,0\ 1\ 1\, \\
    \,1\ 0\ 1\, \\
  \end{pmatrix}. \]
  Thus, $PH=(P,PM,P)$.  Since $P$ is a binary $3\times 3$ invertible
  matrix, the number of ones in~$P$, denoted by $\mcN(P)$, should
  satisfy $\mcN(P)\ge 3$.  If $\mcN(P)\ge 4$, then there exists a
  column in~$P$ such that the number of ones in the column is at least
  two.  Since such a column will occur both in the first block $P$ and
  in the last block $P$ in the matrix $PH=(P,PM,P)$, the cycle can be
  found in these two same columns.  The remaining case is $\mcN(P)=3$,
  and in this case, the block $PM$ in $PH=(P,PM,P)$ is just the
  permutations of the rows of $M$, thus, the block $PM$ contains a
  cycle since $M$ contains a cycle.  These arguments show that $\mcC$
  is not cycle-free.
\end{exa}

\begin{exa}
  Let $k=4$ and $t=3$; consider the code $\mcC$ generated by the
  matrix \[ G =
  \begin{pmatrix}
    \;1\ 0\ 0\ 0\ 1\ 0\ 0\; \\
    \;0\ 1\ 0\ 0\ 1\ 0\ 1\; \\
    \;0\ 0\ 1\ 0\ 0\ 1\ 1\; \\
    \;0\ 0\ 0\ 1\ 1\ 1\ 0\; \\
  \end{pmatrix}. \]

  It can be checked that the first four columns of $G$,
  $p_1,\dots,p_4$, are basis points in $PG(3,2)$, and the last three
  columns, $p_5$, $p_6$, and $p_7$, have representing-sets $S_1 =
  \{p_1,p_2,p_4\}$, $S_2 = \{p_3,p_4\}$, and $S_3 = \{p_2,p_3\}$,
  respectively.  In addition, $U_1 = \{p_1\}$, $U_2 =
  \{p_2,p_3,p_4\}$, and $U_3 = \{p_1, p_5, p_6, p_7\}$.  Furthermore,
  $S_1\cap S_2\cap S_3 = \varnothing$, $S_1\cap U_1 = \{p_1\}$,
  $S_2\cap U_1 = S_3\cap U_1 = \varnothing$, and $m(p_i)=1$ for $1\le
  i\le 7$.  Thus, the points $p_i$, $1\le i\le 7$, satisfy and only
  satisfy the conditions of Case~2) in Theorem~\ref{thm:4}, and
  therefore $\rho(\mcC) = n-k = 7-4 = 3$.

  Similarly to Example~\ref{example}, one may show that $\mcC$ is not
  cycle-free by taking a parity-check matrix $H$ of $\mcC$ as
  \[ H = \begin{pmatrix}
    \,1\ 1\ 0\ 1\ 1\ 0\ 0\, \\
    \,0\ 0\ 1\ 1\ 0\ 1\ 0\, \\
    \,0\ 1\ 1\ 0\ 0\ 0\ 1\, \\
  \end{pmatrix}. \]
\end{exa}

\begin{exa}
  Let $k=4$ and $t=3$; consider the code $\mcC$ generated by the
  matrix \[ G =
  \begin{pmatrix}
    \;1\ 0\ 0\ 0\ 1\ 1\ 1\; \\
    \;0\ 1\ 0\ 0\ 1\ 1\ 1\; \\
    \;0\ 0\ 1\ 0\ 0\ 1\ 1\; \\
    \;0\ 0\ 0\ 1\ 1\ 1\ 0\; \\
  \end{pmatrix}. \]

  Then, the first four columns of $G$, $p_1,\dots,p_4$, are
  the basis points of $PG(3,2)$, and the last three columns of $G$,
  $p_5$, $p_6$, and $p_7$, have representing-sets
  $S_1 = \{p_1,p_2,p_4\}$, $S_2 = \{p_1,p_2,p_3,p_4\}$, and
  $S_3 = \{p_1,p_2,p_3\}$, respectively.  In addition,
  $U_1 = \varnothing$, $U_2 = \{p_1,p_2,p_3,p_4\}$, $U_3 = \{p_5,p_6,p_7\}$,
  $S_1\cap S_2\cap S_3=\{p_1,p_2\}$, and $m(p_i)=1$ for $1\le i\le 7$.
  Thus, the points $p_i$ for $1\le i\le 7$ satisfy and only
  satisfy the conditions of Case 3) in Theorem~\ref{thm:4},
  so that $\rho(\mcC) = n-k = 7-4 = 3$.

  Similarly to the above two examples, one may show that $\mcC$ is not
  cycle-free by taking a parity-check matrix $H$ of $\mcC$ as
  \[ H = \begin{pmatrix}
    \,1\ 1\ 0\ 1\ 1\ 0\ 0\, \\
    \,1\ 1\ 1\ 1\ 0\ 1\ 0\, \\
    \,1\ 1\ 1\ 0\ 0\ 0\ 1\, \\
  \end{pmatrix}. \]
\end{exa}

\begin{rem}
  Along the line of Theorem~\ref{thm:4}, we may use the value
  assignment to get other kinds of codes whose pseudoredundancies can
  be determined, however, the conditions will be too tedious to get
  more information.  So, we omit them.
\end{rem}

Summing up Theorems~\ref{thm:3} and~\ref{thm:4} and using similar
arguments as in these two theorems, we get in general the following
(the detailed proof is omitted):

\begin{thm}
  Let the notations be as before.  Assume that the points $p_{k+i}$,
  $1\le i\le t$, can be divided into $\ell$ subsets such that:
  \begin{enumerate}
  \item the representing-sets of points from different classes do not
    intersect;
  \item each of these $\ell$ classes is either an
    representing-independent one or an representing-dependent one
    satisfying the conditions of Theorem~\ref{thm:4}.
  \end{enumerate}
  Then, the $[n,k]$ code $\mcC$ determined by
  \[ m(p) = \begin{cases}
    z_i\ge 1~ &  \text{if }p=p_i,\ 1\le i\le k+t \,,\\
    0 & \text{otherwise},
  \end{cases} \]
  satisfies $\rho(\mcC) = n-k$ for all the four pseudoweights.
\end{thm}


\section{Conclusion}\label{sec:concl}

Making use of the value assignment, we derived upper bounds on the
pseudoredundancies for certain binary codes with repeated and added
coordinates and for certain shortened subcodes.  Also, we constructed
several kinds of $k$-dimensional binary linear codes by using the
value assignment; the pseudoredundancies for all of
the four pseudoweights of these binary linear codes are fully determined.





\end{document}